\documentclass[11pt, oneside]{article}   	
\usepackage{fullpage}
\usepackage[pass]{geometry}
\usepackage{graphicx}				
\usepackage{amsmath,amssymb,amsthm,enumitem}

\usepackage[bookmarks=true,hypertexnames=false,pagebackref]{hyperref}
\hypersetup{colorlinks=true, citecolor=blue, linkcolor=red, urlcolor=blue}

\usepackage{cleveref}
\usepackage{comment}
\usepackage{tikz}
\usetikzlibrary{arrows,arrows.meta,backgrounds,calc,fit,decorations.pathreplacing,decorations.markings,shapes.geometric}
\tikzset{>=latex}

\newtheorem{theorem}{Theorem}
\newtheorem{lemma}[theorem]{Lemma}

\def\NP{\ensuremath{\mathbf{NP}}}
\def\RP{\ensuremath{\mathbf{RP}}}
\def\Ptime{\ensuremath{\mathbf{P}}}

\def\numP{\textup{\#}\mathbf{P}}

\def\*#1{\mathbf{#1}}
\def\+#1{\mathcal{#1}}
\def\-#1{\mathrm{#1}}
\def\=#1{\mathbb{#1}}

\def\bfz{\mathbf{0}}
\def\bfo{\mathbf{1}}

\def\Rset{\mathbb{R}}
\def\Qset{\mathbb{Q}}
\def\bfS{\mathbf S}
\def\SAT{\textsc{SAT}}
\def\nSAT{\textsc{\#SAT}}
\def\nBIS{\textsc{\#BIS}}
\def\nRet{\textsc{\#Ret}}
\def\nHom{\textsc{\#Hom}}
\def\nPolytopeVertices{\textsc{\#Vertices of a 0/1 Polytope}}
\def\HBIS{\mathcal{H}_{\mathrm{BIS}}}
\def\Hom{\mathop{\mathrm{Hom}}}
\def\eps{\varepsilon}
\newcommand{\abs}[1]{\left\vert#1\right\vert}
\newcommand{\cupdot}{\mathbin{\mathaccent\cdot\cup}}

\let\epsilon=\varepsilon
\let\phi=\varphi
\def\Hfour{H_4}

\title{Approximate counting of vertices of 0/1 polytopes: \\ a stronger hardness result}
\author{Heng Guo\\School of Informatics\\University of Edinburgh\\\url{hguo@inf.ed.ac.uk} \and
Mark Jerrum\thanks{‘This work was supported by the UK Engineering and Physical Sciences Research Council grant number UKRI\,2771.}
\\School of Mathematical Sciences\\Queen Mary, University of London\\\url{m.jerrum@qmul.ac.uk}}
\date{}							

\begin{document}
\maketitle

\begin{abstract}
We show that approximately counting the vertices of a bounded 0/1 polytope, presented as a system of rational linear inequalities, is, informally speaking, $\NP$-hard.   In particular, there is no FPRAS for this problem unless $\RP=\NP$.  The proof is by a reduction from approximately counting homomorphisms from a given graph to a particular four-vertex graph.  
The main proof ideas were found using GPT-5.6 Sol Ultra.

\end{abstract}

\section{Introduction}

We are interested in counting vertices of a polytope defined by linear inequalities $Ax\leq b$.  We concentrate on 0/1 polytopes, namely polytopes all of whose vertices belong to $\{0,1\}^n$.  The polytope is presented by its defining inequalities, rather than by a list of its vertices.
\begin{description}[itemsep=0pt]
\item [Problem:] \nPolytopeVertices.
\item [Instance:] An $m\times n$ rational matrix $A\in\Qset^{m\times n}$ and a vector $b\in\Qset^m$.
\item [Promise:] The inequalities $Ax\leq b$ define a bounded 0/1 polytope~$P$.
\item [Output:] The number $v(P)$ of vertices of~$P$.
\end{description}
It is not known at the time of writing whether there is a polynomial-time algorithm to test whether a given pair $(A,b)$ defines a 0/1-polytope, so we include a promise to this effect.

We allow the polytope to be empty, in which case $v(P)=0$.  An FPRAS for \nPolytopeVertices{} is a randomised algorithm that, given an instance and a parameter $0<\eps\leq1$, outputs a non-negative random value~$Z$ such that
\[
 \Pr\big[(1-\eps)v(P)\leq Z\leq(1+\eps)v(P)\big]\geq\tfrac34,
\]
and runs in time polynomial in the encoding length of $(A,b)$ and in $\eps^{-1}$.

As with the theory of NP-completeness, we may gain evidence for the non-existence of an FPRAS for a given counting problem using a suitable notion of reduction.  In this application, the appropriate notion is \textit{AP-reduction}, short for (polynomial-time) Approximation-Preserving reduction.   In this paper, our reductions are actually \textit{parsimonious}, that is to say, they preserve the number of solutions exactly.  As an efficient parsimonious reduction would meet any reasonable definition of approximation-preserving reduction, we do not provide an exact definition of AP-reduction here.  For a precise definition of concepts used in this paper and a general introduction to the complexity of approximate counting, refer to~Dyer, Goldberg, Greenhill and Jerrum~\cite{Relative}. 

Our main result is the following.
\begin{theorem}\label{thm:main}
\nPolytopeVertices{} is hard for $\numP$ under AP-reductions.  In particular, there is no FPRAS for  \nPolytopeVertices{} unless $\RP=\NP$.
\end{theorem}
A counting problem $\Pi$ is said to ``hard for $\numP$ under AP-reductions'' if, for every counting problem $\Pi'\in\numP$, there is an AP-reduction from $\Pi'$ to $\Pi$.  Some explanatory comments on the statement of this theorem follow later in the section.  For the time being, we merely highlight the conclusion that, under a standard complexity-theoretic assumption, there is no polynomial-time algorithm (even a randomised one) that approximates the number of vertices of a 0/1-polytope.  

The proof of Theorem~\ref{thm:main} is via an AP-reduction from a hard graph homomorphism counting problem.  A homomorphism from a graph~$G$ to a graph~$H$ is a function $\sigma:V(G)\to V(H)$ such that, for all pairs of vertices $u,v\in V(G)$, it is the case that $\{u,v\}\in E(G)$ implies $\{\sigma(u),\sigma(v)\}\in E(H)$.  We denote the set of all homomorphisms from $G$ to $H$ by $\Hom(G,H)$.  For any fixed undirected graph $H$, the homomorphisms-to-$H$ counting problem is defined  as follows.
\begin{description}[itemsep=0pt]
\item [Problem:] $\nHom(H)$.
\item [Instance:] An undirected graph $G$.
\item [Output:]  $|\Hom(G,H)|$, the number of homomorphisms from $G$ to $H$.
\end{description}
In \S\ref{sec:reduction} we present a (deterministic) polynomial-time parsimonious reduction from $\nHom(H)$ to $\nPolytopeVertices$, for the specific graph~$\Hfour$ depicted in Figure~\ref{fig:H}.  As $\nHom(\Hfour)$ is known to be hard for $\numP$ under AP-reductions~\cite{KelkThesis}, Theorem~\ref{thm:main} follows immediately, by transitivity of AP-reducibility.  

As promised, here are some comments on the terminology used in, and the interpretation of, Theorem~\ref{thm:main}.  In the literature on approximate counting, terminology such as ``$\Pi$ is $\nSAT$-hard [under AP-reducibility]'' appears.  Since every problem in $\numP$ is parsimoniously reducible to $\nSAT$ (the problem of counting satisfying assignments of a CNF Boolean formula) these formulations are equivalent to the one used in Theorem~\ref{thm:main}. 

At first sight, the conclusion in Theorem~\ref{thm:main} looks rather weak.  One might imagine that the existence of an FPRAS for  \nPolytopeVertices{} should lead to something more dramatic than $\RP=\NP$, for example to $\RP=\numP$ or the collapse of the polynomial hierarchy.  However, Valiant and Vazirani~\cite{NPeasy} showed that every problem in $\numP$ can be approximated by a polynomial-time randomised Turing machine equipped with an $\NP$-oracle.  So whereas the complexity gap between existence and exact counting is expected to be huge, the gap is almost non-existent for approximate counting.  For the same reason, the phrase ``hard for $\NP$ under AP-reductions'' used in~\cite{GuoJerrum} is equivalent to the formulation used Theorem~\ref{thm:main}, since (with a slight abuse of terminology) every problem in $\numP$ is AP-reducible to any $\NP$-complete problem.  

The standard definition of AP-reduction allows the algorithm performing the reduction to be randomised.  However, the reductions used here are deterministic, so we can equally deduce that there is no polynomial-time deterministic approximation algorithm (technically FPTAS) for \nPolytopeVertices{} unless $\Ptime=\NP$.  (The definition for FPTAS matches the one for FPRAS above, except that the word ``randomised'' is deleted, and the random variable $Z$ becomes a deterministic value.)  Of course, one needs to check that the reductions used to establish hardness results quoted in this paper are deterministic, but that is indeed the case.

A matrix $A$ is \textit{totally unimodular} (TU) if every square submatrix of~$A$ has determinant $-1$, 0 or 1.  It is not difficult to verify that the vertices of the polytope defined by $Ax\leq b$ and $\bfz\leq x\leq\bfo$, where $A$ is TU, lie in $\{0,1\}^n$.  
Many 0/1-polytopes appearing in combinatorial optimisation arise in this way.   One such is the bipartite independent set polytope.  It follows that, restricted to TU instances, approximating $\nPolytopeVertices$ is as hard as approximating $\nBIS$, the problem of counting independent sets in a bipartite graph. No FPRAS is known for $\nBIS$ and the same holds for the wide range of counting problems interreducible with $\nBIS$ via AP-reductions.  At the same time, there is currently no evidence that $\nBIS$ is hard for $\numP$ under AP-reductions.  

As well as observing that approximating $\nPolytopeVertices$ can be as hard as $\nBIS$, Guo and Jerrum~\cite{GuoJerrum} showed that the generalisation to polytopes with vertices in $\{0,\frac12,1\}^n$ is hard for $\numP$ under AP-reductions.  Theorem~\ref{thm:main} is a common strengthening of these two results.  

Two prominent classes of TU matrices are \textit{network matrices} and their transposes.  As far as counting problems arising from transposes of network matrices are concerned we have a clear understanding:  $\nBIS$ is an example in this class, and it is a hardest such under AP-reducibility.  For network matrices our understanding is incomplete, and it is possible that all instances arising from these have an FPRAS~\cite{GuoJerrum}.  We observe in the final section of this paper that the polytope $Q(G)$ used in the main reduction is not TU\null. Thus the approximation complexity of $\nPolytopeVertices$ when restricted to TU instances remains open.  It could be $\nBIS$-equivalent or maximally hard (or somewhere between these). 

A related conjecture by Mihail and Vazirani \cite{MV89} is that the edge expansion of any $0/1$-polytope is at least $1$,
which implies rapid mixing of the simple random walks on the graph of the polytope.
Our hardness result does not directly invalidate this conjecture, since the degrees of vertices can be exponentially large in the dimension of the ambient space.
Nonetheless, this conjecture is recently found false by Yang \cite{Yan26}, who constructed $0/1$-polytopes with exponentially small edge expansion.

\subsection{Statement of AI use}
The main proof ideas were found using GPT-5.6 Sol Ultra.
In particular, it generated a hardness proof from $\nRet(\Hfour)$, which is shown hard in \cite{FGZ21}.
We will first present a reduction from $\nHom(\Hfour)$, which is simpler and the hardness of $\nHom(\Hfour)$ is shown by Kelk \cite{KelkThesis} in his PhD thesis.
The proofs are written or rewritten by the authors, who take responsibility of any error.

\section{The reduction}\label{sec:reduction}

Our reduction rests on the graph $\Hfour$ with vertex set $V(\Hfour)=\{00,10,01,11\}$ depicted in \Cref{fig:H}. Note that we identify every vertex of~$\Hfour$ with a vector in $\{0,1\}^2$.  For $c,d\in V(\Hfour)$, including the case $c=d$, we put $\{c,d\}$ in the edge set $E(\Hfour)$ if and only if
\begin{equation}\label{eq:H-adjacency}
 \abs{c}_1+\abs{d}_1\geq2,
\end{equation}
where $\abs{c}_1$ is the Hamming weight of~$c$.  Thus $10$, $01$ and $11$ form a reflexive triangle (i.e., a triangle with loops on all vertices), while $00$ is an unlooped leaf adjacent only to~$11$.

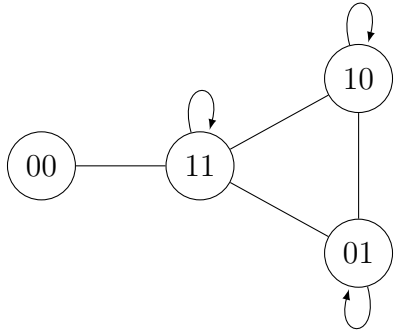
\begin{figure}[htbp]
  \centering
  \begin{tikzpicture}[scale=1.05, every node/.style={font=\large}]
    \node[draw,circle,minimum size=9mm] (z) at (0,0) {$00$};
    \node[draw,circle,minimum size=9mm] (o) at (2,0) {$11$};
    \node[draw,circle,minimum size=9mm] (a) at (4,1.1) {$10$};
    \node[draw,circle,minimum size=9mm] (b) at (4,-1.1) {$01$};
    \draw (z)--(o)--(a)--(b)--(o);
    \draw (a) edge[loop above] (a);
    \draw (b) edge[loop below] (b);
    \draw (o) edge[loop above] (o);
  \end{tikzpicture}
  \caption{The graph~$H_4$.}
  \label{fig:H}
\end{figure}

\begin{lemma}\label{lem:reduction}
$\nHom(\Hfour)$ is AP-reducible to $\nPolytopeVertices$.
\end{lemma}

\begin{proof}
Let $G$ be an instance of $\nHom(\Hfour)$.  
We construct a polytope $Q(G)$ in a $2|V(G)|$-dimensional ambient space.  Let $y\in\mathbb{R}^{V(G)\times\{1,2\}}$ be a vector,
and for every $v\in V(G)$ write
\[
 s_v=y_{v,1}+y_{v,2}.
\]
Define the polytope $Q(G)\subseteq\Rset^{2\abs{V(G)}}$ as the set of vectors $y$ satisfying
\begin{align}
 0\leq y_{v,i}\leq1, &\quad\text{for all $v\in V(G)$ and $i\in\{1,2\}$}; \label{eq:box}\\
 s_u+s_v\geq2,&\quad\text{for all $\{u,v\}\in E(G)$}. \label{eq:edge}
\end{align}
Note that $(s_v)$ are auxiliary variables introduced for convenience in the proof, so \eqref{eq:edge} should be interpreted as a shorthand for 
$$
y_{u,1}+y_{u,2}+y_{v,1}+y_{v,2}\geq2.
$$
We aim to show that the vertices of $Q(G)$ are in bijection with homomorphisms from $G$ to~$\Hfour$.  This gives the required AP-reduction from $\nHom(\Hfour)$ to $\nPolytopeVertices$.
The key observation is that every vertex of $Q(G)$ belongs to $\{0,1\}^{2\abs{V(G)}}$.

Suppose, for a contradiction, that $y$ is a vertex of $Q(G)$ with a fractional coordinate.  We call $(y_{v,1},y_{v,2})$ the \textit{block} corresponding to~$v$.
We first claim that every fractional block has exactly one fractional coordinate.  Indeed, suppose both coordinates in the block corresponding to~$v$ are fractional.  For all sufficiently small $\delta>0$, we can add $\delta$ to one coordinate and subtract $\delta$ from the other.  This preserves $s_v$, and hence preserves all the inequalities in \eqref{eq:edge}.  Both directions of the perturbation satisfy the box constraints in \eqref{eq:box}.  This contradicts the assumption that $y$ is a vertex.  The claim follows.  In particular, $s_v$ is non-integral whenever the block corresponding to~$v$ is fractional.

Let $F\subseteq V(G)$ be the set of vertices whose blocks are fractional.  Define a graph~$T$ on vertex set~$F$ by putting $\{u,v\}\in E(T)$ if $\{u,v\}\in E(G)$ and the corresponding inequality is tight, namely
\begin{equation}\label{eq:tight}
 s_u+s_v=2.
\end{equation}
Note that a tight inequality cannot have exactly one endpoint in~$F$, since the sum of an integer and a non-integer cannot be~2.

We claim that $T$ is bipartite.  Otherwise, let $v_0v_1\cdots v_{2k}v_0$ be an odd cycle in~$T$.  The equations in \eqref{eq:tight} imply
\[
 s_{v_0}=s_{v_2}=\cdots=s_{v_{2k}}
 \qquad\text{and}\qquad
 s_{v_{2k}}=2-s_{v_0}.
\]
The final edge $v_{2k}v_0$ gives $2s_{v_0}=2$.  Hence $s_{v_0}=1$, contradicting the fact that $s_{v_0}$ is non-integral.

Now let $C$ be a connected component of~$T$, with bipartition $C=C^+\cupdot C^-$.  For every $v\in C$, perturb the unique fractional coordinate in its block by $+\delta$ if $v\in C^+$, and by $-\delta$ if $v\in C^-$.  The opposite perturbation is obtained by reversing the signs.  Every tight edge inequality incident with~$C$ has both endpoints in~$C$ and is preserved.  Every other affected edge inequality has positive slack.  Also, all the coordinates being changed lie strictly between 0 and~1.  Thus both perturbations are feasible for all sufficiently small $\delta>0$, again contradicting the assumption that $y$ is a vertex.  It follows that $F$ is empty, and hence every vertex of $Q(G)$ is binary.

It remains to show that the reduction mapping $G$ to $Q(G)$ is parsimonious.  Observe that there is a natural bijection $\phi:\{0,1\}^{V(G)\times\{1,2\}}\to V(\Hfour)^{V(G)}$ that interprets each binary vector $y\in\{0,1\}^{V(G)\times\{1,2\}}$ as a function $\sigma_y:V(G)\to V(\Hfour)$.  Explicitly, for each vertex $v\in V(G)$ we define $\sigma_y(v)=w$, where $w$ is the vertex of~$\Hfour$ encoded by the pair $(y_{v,1},y_{v,2})$.

In one direction, suppose that $y$ is a vertex of $Q(G)$.  We saw that $y$ is binary vector, so $\phi(y)$ is a function from $V(G)$ to $V(\Hfour)$.  It is easy to see that the construction of $Q(G)$ --- specifically the set of inequalities~\eqref{eq:edge} --- forces $\phi(y)$ to be a homomorphism.  This gives an injective map from vertices of $Q(G)$ to $\Hom(G,\Hfour)$.  In the other direction, any homomorphism  $\sigma\in\Hom(G,\Hfour)$ corresponds to a binary vector $y=\phi^{-1}(\sigma)$ and the construction of $Q(G)$ ensures that $y\in Q(G)$.  Any binary vector lying in a 0/1-polytope is a vertex of that polytope.  Thus the function $\phi^{-1}$ is an injective map from $\Hom(G,\Hfour)$ to vertices of $Q(G)$.
\end{proof}  

\begin{proof}[Proof of Theorem~\ref{thm:main}]
$H_4$ is \textit{Graph 39} in Kelk's PhD thesis~\cite{KelkThesis}. On p.71 he shows that $\nSAT$ is AP-reducible to $\nHom(\Hfour)$ and hence $\nHom(\Hfour)$ is hard for $\numP$ under AP-reductions. This fact, combined with Lemma~\ref{lem:reduction} and transitivity of AP-reductions, proves the main claim in Theorem~\ref{thm:main}.

For the subsidiary claim, take any NP-complete problem, say $\SAT$, and consider its counting analogue $\nSAT$.  Suppose, for a contradiction, that $\nPolytopeVertices$ has an FPRAS.  Let $\Psi$ (a CNF formula) be an instance of $\SAT$ and hence also of $\nSAT$.  By the first claim, there is an AP reduction from $\nSAT$ to $\nPolytopeVertices$.  This, combined with the FPRAS for  $\nPolytopeVertices$, yields an FPRAS for $\nSAT$.  Setting $\epsilon=\frac12$ in the definition of FPRAS, we have a randomised algorithm that, with error probability $\frac14$, can decide whether $\Psi$ has no satisfying assignments or a positive number.  Thus we have a polynomial-time randomised algorithm for $\SAT$ with two-sided error probability.  A standard technique~(see e.g., \cite[Ex 1.15]{MR95}) strengthens this to a randomised algorithm with one-sided errors.  So $\SAT\in\RP$ and $\RP=\NP$.
\end{proof}

A possible objection to the above line of argument is that it is based on a result from a non-peer-reviewed source.  We can avoid this objection at the expense of a little extra work.

The retraction counting problem is the following. 
\begin{description}[itemsep=0pt]
\item [Problem:] $\nRet(H)$.
\item [Instance:] An undirected graph $G$ without loops, and a collection of sets $\bfS=\{S_v\subseteq V(H):v\in V(G)\}$ such that, for all $v\in V(G)$, we have $|S_v|\in\{1,|V(H)|\}$.  
\item [Output:] The number of homomorphisms $\sigma$ from $G$ to $H$ such that, for all $v\in V(G)$, we have $\sigma(v)\in S_v$.
\end{description}

If we think of the vertices of $H$ as being colours, and a function $\sigma:V(G)\to V(H)$ as being a colouring of the vertices of $G$, then $\nRet(H)$ counts ``$H$-colourings'' of $G$ in which certain vertices of $G$ are forced to take on a specified colour.\footnote{This not a ``retraction'' as usually defined, but equivalent in this context, as explained by Focke et al.~\cite{FGZ21}.} For the particular $\Hfour$ in Figure~\ref{fig:H}, we have the following analogue of Lemma~\ref{lem:reduction}

\begin{lemma}\label{lem:variantreduction}
$\nRet(\Hfour)$ is AP-reducible to $\nPolytopeVertices$.
\end{lemma}

This is a slightly stronger version of the previous lemma.
The advantage of this version is that the hardness proof for $\nRet(\Hfour)$ appears in archival form~\cite{FGZ21}.  The proof of Lemma~\ref{lem:variantreduction} is a mild complication of the proof of Lemma~\ref{lem:reduction}.

\begin{proof}[Proof of Lemma~\ref{lem:variantreduction}]\label{lem:variantreduction'}
Given an instance $(G,\bfS)$ of $\nRet(\Hfour)$ we code retractions as vectors $y\in\mathbb{Q}^{V(G)\times\{1,2\}}$ as before.  The polytope $Q(G,\bfS)$ is defined by the same linear inequalities as in the earlier proof, together with some additional equalities coding the sets $S_v\in\bfS$. For example, the set $S_v=\{01\}$ is coded by the equalities $y_{v,1}=0$ and $y_{v,2}=1$.  No additional equalities are added when $S_v=V(\Hfour)$.

Consider a vertex $y$ of $Q(G,\bfS)$. For the same reason as in the earlier proof, it cannot be the case that $y_{v,1}$ and $y_{v,2}$ are both fractional.  The blocks and tight edge constraints induce a bipartite subgraph as before.  Noting that none of these blocks correspond to vertices~$v$ with $|S_v|=1$, the rest of the argument goes through as before.  We deduce that the vertices of $Q(G,\bfS)$ are in bijection with retractions to~$\Hfour$.
\end{proof}

To complete the alternative proof of Theorem~\ref{thm:main} we just need to compare $\Hfour$ against the complexity classification of Focke, Goldberg and \v{Z}ivn\'y~\cite[Theorem~2]{FGZ21}.  
The underlying simple graph of~$\Hfour$ is a triangle with a pendant vertex.  In particular, $\Hfour$ is connected and square-free.  It is enough to check that $\Hfour$ is not in any of the easy or $\nBIS$-equivalent cases in that theorem. 

The graph $\Hfour$ contains both looped and unlooped vertices.  Thus it is neither a reflexive clique nor an irreflexive complete bipartite graph.  Also, since $\Hfour$ has loops, it is not an irreflexive caterpillar.

It remains to rule out the class $\HBIS$.  Suppose, for a contradiction, that $\Hfour\in\HBIS$.  In the notation of~\cite[Definition~10]{FGZ21}, there is a positive integer~$m$, reflexive cliques $K_0,\ldots,K_m$, and distinct looped vertices $p_0,\ldots,p_{m+1}$ such that $p_i,p_{i+1}\in K_i$ and
\[
 K_{i-1}\cap K_i=\{p_i\}.
\]
Since $\Hfour$ has exactly three looped vertices, we must have $m=1$.  The two cliques $K_0$ and $K_1$ intersect only at~$p_1$, so $p_0p_2$ is not an edge of~$\Hfour$.  This is a contradiction, since the three looped vertices of~$H_4$ form a triangle.

The third case of the trichotomy applies, and $\nRet(\Hfour)$ is $\nSAT$-hard (and hence $\numP$-hard) under approximation-preserving reductions.

\section{Total unimodularity}
Recall that a matrix is totally unimodular (TU) if every square submatrix has determinant $-1$, 0 or 1.   It is not to difficult to see that the polytopes $Q(G)$ constructed in the main reduction do not necessarily correspond to TU constraint matrices.  In particular, consider the linear inequalities defining the polytope $Q(C_3)$, where $C_3$ is the cycle on vertex set $\{u,v,w\}$.  They include the inequalities
$$\begin{matrix}
y_{u,1}+y_{u,2}&{}+y_{v,1}+y_{v,2}&&\geq2\\
y_{u,1}+y_{u,2}&&{}+y_{w,1}+y_{w,2}&\geq2\\
&\hphantom{{}+} y_{v,1}+y_{v,2}&{}+y_{w,1}+y_{w,2}&\geq2
\end{matrix}$$
together with the box constraints.  So, written in the form $Ay\leq b$, the matrix $A$ contains the $3\times6$ submatrix
$$
-\begin{pmatrix}
1&1&1&1&0&0\\
1&1&0&0&1&1\\
0&0&1&1&1&1
\end{pmatrix}
$$
which in turn contains the $3\times3$ submatrix 
$$-\begin{pmatrix}
1&1&0\\
1&0&1\\
0&1&1
\end{pmatrix}
$$
As this matrix has determinant~2, the matrix $A$ is not totally unimodular.

It is possible to speculate that escaping from total unimodularity is critical to the proof of Theorem~\ref{thm:main}.

\bibliographystyle{plain}
\bibliography{polytope-vertex-count}

\end{document}